\documentclass[12pt]{article}
\usepackage{graphicx}

\usepackage{amsmath}
\usepackage{amsthm}
\usepackage{amsfonts}
\usepackage{amssymb}
\usepackage[english]{babel}

\title{Economic Cycles of Carnot Type}
\author{C. Udriste, V. Golubyatnikov, I. Tevy}
\date{Bicentennial of the University Politehnica of Bucharest 1818-2018}

\begin{document}
\maketitle
\textheight 19cm
\textwidth 16cm
\topmargin 0,1cm
\newcommand{\di}{\displaystyle}
\newcommand{\ov}{\over}
\newcommand{\noa}{\noalign{\medskip}}
\newcommand{\al}{\alpha}
\newcommand{\be}{\beta}
\newcommand{\om}{\omega}
\newcommand{\bt}{\bar\tau}
\newcommand{\br}{\hbox{\bf R}}
\newcommand{\bd}{\hbox{\bf D}}
\newcommand{\bp}{\hbox{\bf P}}
\newcommand{\bn}{\hbox{\bf N}}
\newcommand{\bee}{\hbox{\bf E}}
\newcommand{\bc}{\hbox{\bf C}}
\newcommand{\bb}{\hbox{\bf B}}
\newcommand{\bv}{\hbox{\bf V}}
\newcommand{\sgn}{\hbox{sgn}}
\newcommand{\rot}{\hbox{rot}}
\newcommand{\divv}{\hbox{div}}
\newcommand{\bh}{\hbox{\bf H}}
\newcommand{\bm}{\hbox{\bf M}}
\newcommand{\ome}{\Omega}
\newcommand{\ba}{\bar a}
\newcommand{\ti}{\times}
\newcommand{\ga}{\gamma}
\newcommand{\ty}{\infty}
\newcommand{\de}{\delta}
\newcommand{\te}{\theta}
\newcommand{\na}{\nabla}
\newcommand{\pa}{\partial}
\newcommand{\va}{\varphi}
\newcommand{\ld}{\ldots}
\newcommand{\qu}{\quad}
\newcommand{\la}{\lambda}
\newcommand{\fo}{\forall}
\newcommand{\ep}{\varepsilon}
\newcommand{\pp}{\prime}
\newcommand{\su}{\subset}
\newcommand{\si}{\sigma}
\newcommand{\dd}{\Delta}
\newcommand{\gaa}{\Gamma}
\newcommand{\ri}{\Rightarrow}
\newcommand{\rii}{\Leftrightarrow}
\newcommand{\med}{\medskip}

\theoremstyle{definition}
\newtheorem{definition}{Definition}[section]
\theoremstyle{theorem}
\newtheorem{theorem}{Theorem}[section]
\newtheorem{proposition}{Proposition}[section]
\newtheorem{corollary}[theorem]{Corollary}
\newtheorem{lemma}[theorem]{Lemma}
\newtheorem{remark}{Remark}[section]

\begin{abstract}
Originally, the Carnot cycle is a theoretical thermodynamic cycle that
provides an upper limit on the efficiency that any classical thermodynamic engine can
achieve during the conversion of heat into work, or conversely, the efficiency of a
refrigeration system in creating a temperature difference by the application of work
to the system.

The aim of this paper is to introduce and study the economic Carnot cycles
into a Roegenian economy, using our Thermodynamic-Economic Dictionary.
Of course, the most difficult questions are: what is the economic significance of such a cycle?
Roegenian economics is acceptable or not, in terms of practical applications?
Our answer is yes for both questions.
\end{abstract}

{\bf AMS Mathematical Classification}: 80A99, 91B54, 91B74

{\bf J.E.L. Classification: B41}


\section{Thermodynamic-Economic Dictionary}

Thermodynamics is important as a model of phenomenological theory, which describes and
unifies certain properties of different types of physical systems. There are many
systems in biology, economics and computer science, for which an organization
similar and unitary-phenomenological would be desirable. Our purpose is to present
certain features of the economy that are inspired by thermodynamics and vice versa.
In this context, we offer again a Dictionary that reflects the Thermodynamics-Economy isomorphism.
The formal analytical-mathematical analogy between economics and thermodynamics is now well-known
or at least accepted by economists and physicists as well.

Starting from these remarks, the works of Udriste et al. \cite{[10]}-\cite{[20]}
build an isomorphism between thermodynamics and the economy, admitting that
fundamental laws are also in correspondence through our identification.
Therefore, each thermodynamic system is naturally equivalent to an economic system,
and thermodynamic laws have correspondence in the economy.

In the following, we reproduce again the correspondence between the characteristic
state variables and the laws of thermodynamics with the macro-economics ingredients
as described in Udriste et al. (2002-2013) based on the theory on which the economy is founded in 1971
\cite{[3]}. They also allow us to include in the economy the idea of a "black hole"
with a similar meaning to the one in astrophysics [Udriste, Ferrara \cite{[16]}].
We do not know if Roegen would have judged this, but that's why now we are judging him instead.

\vspace{1cm}
{\hspace{-1cm}
\begin{tabular}{lll}
\vspace{0.3cm}
\hspace{0.5cm}THERMODYNAMICS&\hfill & \hspace{0.7cm}ECONOMICS \\
U=\hbox{internal energy} & \hfill\ldots &G=\hbox{growth potential}\\
T=\hbox{temperature}&\hfill\ldots& I=\hbox{internal politics stability}\\
S=\hbox{entropy}&\hfill \ldots& E=\hbox{entropy}\\
P=\hbox{pressure}& \hfill \ldots& P=\hbox{price level (inflation)}\\
V=\hbox{volume}&\hfill \ldots& Q=\hbox{volume, structure, quality}\\
M=\hbox{total energy (mass)}&\hfill \ldots& Y=\hbox{national income (income)}\\
Q=\hbox{electric charge}& \hfill \ldots& $\mathcal{I}$=\hbox{total investment}\\
J= \hbox{angular momentum}&\hfill \ldots& J=\hbox{economic angular momentum}\\
\hspace{1cm}\hbox{(spin)}&\hfill &\hspace{1cm}\hbox{(economic spin)}\\
M=M(S,Q,J)& \hfill \ldots& Y=Y(E,{$\mathcal{I}$},J)\\
$\Omega = \frac{\partial M}{\partial J}$= \hbox{angular speed}&\hfill \ldots& $\frac{\partial Y}{\partial J}$=\hbox{marginal inclination to rotate}\\
$\Phi = \frac{\partial M}{\partial Q}$=\hbox{electric potential}&\hfill \ldots& $\frac{\partial Y}{\partial {\cal{I}}}$=\hbox{marginal inclination to investment}\\
$T_H= \frac{\partial M}{\partial S}$=\hbox{Hawking temperature}&\hfill \ldots& $\frac{\partial Y}{\partial E}$=\hbox{marginal inclination to entropy}\\

\end{tabular}}
\vspace{1cm}

{\it The Gibbs-Pfaff fundamental equation in thermodynamics}
$ dU-TdS + PdV + \sum_k \mu_k dN_k = 0 $ is changed to
{\it Gibbs-Pfaff fundamental equation of economy}
$ dG-IdE + PdQ + \sum_k \nu_k d {\cal {N}} _ k = 0 $.
These equations are combinations of the first law and the
second law (in thermodynamics and economy respectively).
The third law of thermodynamics $ \di \lim_{T \to 0} S = 0 $
suggests the third law of economy $ \di \lim_ {I \to 0} E = 0 $
"if the internal political stability $ I $ tends to $ 0 $,
the system is blocked, meaning entropy becomes $ E = 0 $,
equivalent to maintaining the functionality of the economic system must cause disruption").

Process variables $ W = $ {\it mechanical works} and $ Q $ = {\it heat}
are introduced into elementary mechanical thermodynamics by $ dW = PdV $
(the first law) and by elementary heat, respectively, $ dQ = TdS $, for reversible processes,
or $ dQ <TdS $, for irreversible processes ({\it second law}).
Their correspondence in the economy,
$$ W = \hbox {\it wealth of the system},\,\,q= \hbox{\it production of goods}$$
are defined by $ dW = Pdq $ ({\it elementary wealth in the economy})
and $ dq = IdE $ or $ dq <IdE $ ({\it the second law or the elementary
production of commodities}). A commodity
is an economic good, a product of human labor, with a utility in the sense of life,
for sale-purchase on the market in the economy.

Sometimes a thermodynamic system is found in an
{\it external electromagnetic field} $ (\vec {E}, \vec {H}) $.
{\it The external electric field} $ \vec E $ determines
the {\it polarization} $ \vec P $ and {\ it the external magnetic field} $ \vec H $
determine {\it magnetizing} $ \vec M $. Together they give the
total elementary mechanical work  $ dW = P\,dV + \vec E\, d \vec P + \vec H\, d \vec M $.
Naturally, an economic system is found in an
{\it external econo-electromagnetic field} $(\vec {e}, \vec {h})$.
The {\it external investment (econo-electric) field} $\vec e$ determines
{\it initial growth condition field (econo-polarization field)} $ \vec p $
and the {\it external growth field (econo-magnetic field)} $ \vec h $
cause {\it growth (econo-magnetization)} $ \vec m $. All these fields
produce the elementary mechanical work $dW=P\,dQ+\vec e\,
d\vec e +\vec h\, d\vec m$. The economic fields introduced here are
imposed on the one hand by the type of economic system and on the
other hand by the policy makers (government, public companies, private firms, etc.).

The long term association between Economy and Thermodynamics can be strengthened
with new tools based on the previous dictionary. Of course, this new idea of
the thermodynamically-economical dictionary produces concepts different from those in econophysics \cite{[1]}
and concepts of Thermodynamics in economic systems \cite{[7]}, \cite{[JM]}, \cite{[9]}.
Econophysics seems to build similar economic notions to physics, as if those in the economy were not enough.

The thermodynamic-economic dictionary allows the transfer of information from
one discipline to another (Udriste et al., \cite{[10]}-\cite{[20]}),
keeping the background of each discipline, that we think that was suggested by Roegen in 1971 \cite{[3]}.

\begin{definition}
{\it Economics based on rules similar to those in thermodynamics is called Roegenian economics}.
\end{definition}

Economics described by Gibbs-Pfaff equation is Roegenian economics.

A macro-economic system based on a Gibbs-Pfaff equation is both Roegenian and controllable (see \cite{[19]}).

Roegenian economics (also called bioeconomics by Georgescu-Roegen) is both a
transdisciplinary and an interdisciplinary field of academic research addressing
the interdependence and coevolution of human economies and natural ecosystems,
both intertemporally and spatially. The economic entropy is closely related to
disorder and complexity \cite{[4]}.

\section{What Is the Carnot Cycle in \\Thermodynamics?}

The Carnot cycle \cite{[2]}, \cite{[BDU]} is a theoretical thermodynamic cycle proposed by
French physicist Sadi Carnot. It provides an upper limit on the efficiency that
any classical thermodynamic engine can achieve during the conversion of
heat into work, or conversely, the efficiency of a refrigeration system
in creating a temperature difference by the application of work to the system.

A heat engine is a device that produces motion from heat and
includes gasoline engines and steam engines. These devices
vary in efficiency. The Carnot cycle describes the most efficient
possible heat engine, involving two isothermal processes and
two adiabatic processes. It is the most efficient heat engine
that is possible within the laws of physics.

The second law of thermodynamics states that it is impossible to
extract heat from a hot reservoir and use it all to do work;
some must be exhausted into a cold reservoir. Or, in other words,
no process can be "one hundred percent"
efficient because energy is always lost somewhere. The Carnot cycle
sets the upper limit for what is possible, for what the maximally-efficient
engine would look like.

When we go through the details of the Carnot cycle, we should
define two important adjectives, isothermic and adiabatic, which characterize a process.
An isothermic process means that the temperature remains constant and the
volume and pressure vary relative to each other. An adiabatic process
means that no heat enters or leaves the system to or from a
reservoir and the temperature, pressure, and volume are all free to change,
relative to each other.

\begin{remark}
The idea of Carnot's cycles in the Roegenian economy was suggested by
Vladimir Golubyatnikov in a discussion held at the "The XII-th International Conference
of Differential Geometry and Dynamical Systems (DGDS-2018), 30 August - 2 September 2018, Mangalia, Romania"  on
the communication of Constantin Udriste about the Roegenian economy and
the dictionary that creates this point of view in the economy.

The economic Carnot cycle in Roegenian economics, introduced and analysed in the next two Sections,
are and are not the same with economic cycles described in \cite{[BDU]}, \cite{[41]}, \cite{[AKh]},  \cite{[AK]}.
\end{remark}

\section{Economic Carnot Cycle, ${\bf Q-P}$ Diagram}

Let us consider again the Roegenian economics (an economic distribution) generated by the Gibbs-Pfaff fundamental
equation $dG-IdE + PdQ =0$ on $\mathbb{R}^5$, where $(G, I, E, P, Q)$ are economic variables, according the previous dictionary
($G$= growth potential, $I$= internal politics stability, $E$=entropy, $P$= price level (inflation),
$Q$= volume, structure, quality).

The heat $Q$ in thermodynamics is associated to $q = \hbox{\it production of goods}$ in economics.
Let $I$ be the internal politics stability, $P$ be the price level, and $E$ the economic entropy.

The properties of an economic cycle are represented firstly on the ${\bf Q-P}$ diagram (Figure 1, arrows clockwise).
This diagram is based on economical quantities $Q$ and $P$ that we can measure.
In the ${\bf Q-P}$ diagram, the curves (evolutions) $1-2$, $3-4$
are considered at constant internal politics stability (corresponding to isothermic transformations)
and the curves (evolutions) $2-3$, $4-1$ are considered at constant levels of
goods production regarding as supply (corresponding to adiabatic transformations).
See, "supply and demand law" in economics.

Let us describe the four steps that make up an economic Carnot cycle.
Process $1-2$ is an iso-internal politics stability expansion,
where the volume increases and the price level decreases at a constant internal politics stability.
Process $2-3$ is a constant levels of goods production expansion.

\begin{figure}
  \centering
  \includegraphics[width=9cm]{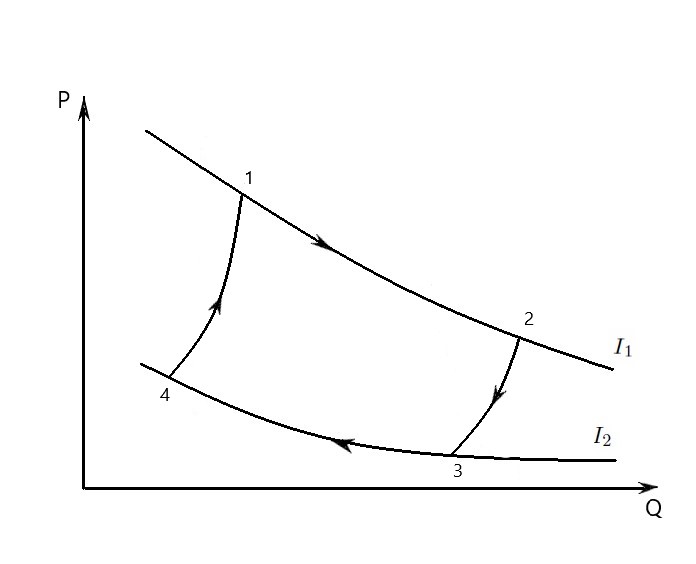}\\
 \caption{Economic Carnot cycle illustrated on a Q-P diagram to show the wealth done}\label{ECC}
 \end{figure}

\begin{figure}
  \centering
  \includegraphics[width=9cm]{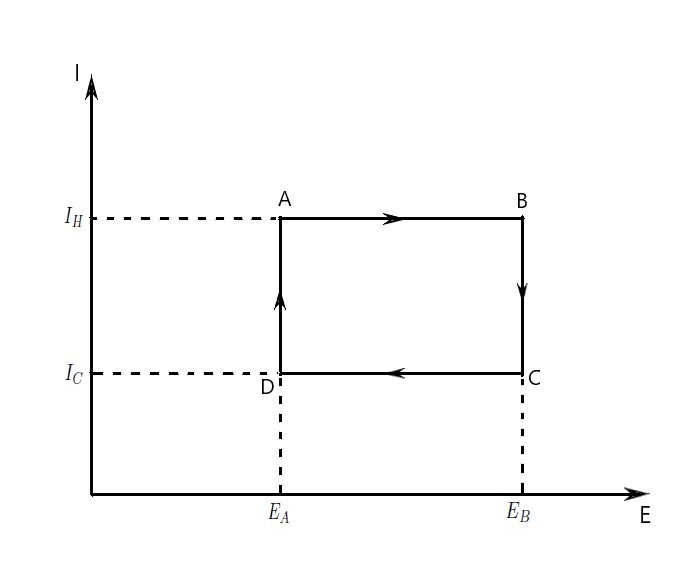}\\
  \caption{Economic Carnot cycle acting as economic engine producing goods, illustrated on a E-I diagram}\label{ECC}
\end{figure}

Process $3-4$ is an iso-internal politics stability compression,
where the volume decreases and the price level increases at an internal politics stability.
Lastly, process $4-1$ is a constant levels of goods production compression.
Process $2-3$ is a constant levels of goods production expansion.
Process $3-4$ is an iso-internal politics stability compression,
where the volume decreases and the price level increases at an internal politics stability.
Lastly, process $4-1$ is a constant levels of goods production compression.

The useful wealth $W$ in the economy that comes out of the economic Carnot cycle is the
difference between the wealth $W_1$ in the economy done by the economic system in stages $1-2$ and $2-3$
and the wealth $W_2$ in the economy done (or the economic energy wasted) by you in stages $3-4$ and $4-1$.
The economic Carnot cycle described above is the most favorable case, because it
produces the largest difference between these values allowed by the laws of economics.

\section{Economic Carnot Cycle, ${\bf E-I}$ Diagram}

Let $(G, I, E, P, Q)$ be economic variables, according the previous dictionary
($G$= growth potential, $I$= internal politics stability, $E$=entropy, $P$= price level (inflation),
$Q$= volume, structure, quality). We refer again the Roegenian economics (an economic distribution)
generated by the Gibbs-Pfaff fundamental equation $dG-IdE + PdQ =0$ on $\mathbb{R}^5$.

The heat $Q$ in thermodynamics is associated to $q= \hbox{\it production of goods}$ in economics.
Let $I$ be the internal politics stability, and $E$ the economic entropy satisfying the Pfaff equation $dq=IdE$.
The properties of an economic process are represented secondly on the ${\bf E-I}$ diagram (Figure 2),
although $E$ and $I$ are not directly measurable.

The evolution in the plane ${\bf E-I}$ of an economic system in time will be described by a curve
connecting an initial state (A) and a final state (B). The area under the curve will be
$${\displaystyle q=\int _{A}^{B}I\,dE},$$
which is the amount of production of goods (economical energy) transferred in the process. If the process
moves to greater entropy, the area under the curve will be the amount of production of goods
absorbed by the system in that process. If the process moves towards lesser entropy,
it will be the amount of production of goods removed. For any cyclic process, there will be an
upper portion of the cycle and a lower portion. For a clockwise cycle, the area
under the upper portion will be the economical energy absorbed during the cycle,
while the area under the lower portion will be the economical energy removed during
the cycle.

The area inside the cycle will then be the difference between the two,
but since the economical energy of the system must have returned to its initial value,
this difference must be the amount of the wealth of the economic system over the cycle.

\begin{theorem}
The amount of the wealth of the economic system done over a cyclic process is
$${\displaystyle W= \oint IdE = (I_{H}-I_{C})(E_{B}-E_{A})}.$$
\end{theorem}

\begin{proof}

Referring to "$W$ = the wealth of the system, $P$=price level (inflation), $Q$ = volume,
structure, quality", mathematically, for a reversible process we may write
the amount of the wealth of the economic system done over a cyclic process as
$${\displaystyle W=\oint PdQ=\oint dq-dG =\oint IdE-dG = \oint IdE-\oint dG=\oint IdE}.$$
Since $dG$ is an exact differential, its integral over any closed loop is zero
and it follows that the area inside the loop on a ${\bf E-I}$ diagram is equal to the
total amount of the wealth performed if the loop is traversed in a clockwise direction,
and is equal to the total amount of the wealth done on the system as the loop is traversed
in a counterclockwise direction.

Evaluation of the above integral is particularly simple for the Carnot cycle if we use
the ${\bf E-I}$ diagram (Figure 2).
The amount of economic energy transferred as wealth of the system is (see Green formula)
$${\displaystyle W=\oint PdQ=\oint IdE=\iint_\Delta dI\,dE =(I_{H}-I_{C})(E_{B}-E_{A})}.$$

\end{proof}

The total amount of economic energy transferred from the production
to the system is given by
${\displaystyle q_{H}=I_{H}(E_{B}-E_{A})}$,
and the total amount of economic energy transferred from the system to the consumption is
$${\displaystyle q_{C}=I_{C}(E_{B}-E_{A})},$$
where ${\displaystyle E_{B}}$ is the maximum system entropy,
${\displaystyle E_{A}}$ is the minimum system entropy.

\section{Efficiency of an Economic System}

We recall the notations: ${\displaystyle W}$ is the wealth done by the system (economic energy),
${\displaystyle q_{C}}$ is the production of goods taken from the system (economic energy leaving the system),
${\displaystyle q_{H}}$ is the production of goods put into the system (economic energy entering the system),
${\displaystyle I_{C}}$ is the absolute internal politics stability of the consumption,
${\displaystyle I_{H}}$ is the absolute internal politics stability of the production.

\begin{definition}
The number
$${\displaystyle \eta ={\frac {W}{q_{H}}}=1-{\frac {I_{C}}{I_{H}}}}$$
is called the {\it efficiency} of an economic system.
\end{definition}
The absolute internal politics stability of the consumption can not be as small as
it is limited by social restrictions. The absolute internal politics stability of the production
can not be as large as it is limited by resources and labor. A wrong interpretation of
the efficiency of an economic Roegenian system can generate
{\bf Homeland Falcons Hymn}: to grow strong and big without eating anything.

This definition of efficiency makes sense for an economic "engine",
since it is the fraction of the economic energy extracted from
the production and converted to wealth done by the system (economic energy).
In practical aspects of this theory, we use approximations.

Economic Carnot cycle has maximum efficiency for reversible economic "engine".

The Carnot economic cycle previously described is a totally reversible cycle.
That is, all the processes that comprise it can be reversed, in which
case it becomes the consumption Carnot cycle. This time, the cycle
remains exactly the same except that the directions of any production and wealth
interactions are reversed. The $\bf{Q-P}$ diagram of the reversed economic Carnot cycle
is the same as for the initial economic Carnot cycle except that the directions of the processes are reversed.

\section{Ideal Income Case}

To transform the ideal gas theory into the ideal income theory,
we recall a part of dictionary and we complete it with new correspondences:
$Q$ = heat $\leftrightarrow$ $q$ {\it production of goods};
$T$ = temperature $\leftrightarrow$ $I$ = {\it internal politics stability}; $P$ = pressure $\leftrightarrow$ $P$ ={\it price level};
$S$ = entropy $\leftrightarrow$ $E$ = {\it economic entropy}; $V$ = volume $\leftrightarrow$ $Q$ ={\it  volume, structure, quality};
$U$ = internal energy $\leftrightarrow$ $G$ = {\it growth potential}; mole $\leftrightarrow$ {\it economic mole};
$R$ = molar gas constant $\leftrightarrow$ $R$ = {\it molar income constant}; $f$ = number of degrees of freedom
$\leftrightarrow$ $f$ = {\it number of degrees of freedom}.
Also, we look at ${\bf Q-P}$ diagram.

Suppose the three states of an economy "inflation,  monetary policy as liquidity, income"
are reduced to one and the same "income".
By the equipartition theorem, the growth potential of one mole of an ideal income is
$G= \frac{f}{2}\,R\,I$, where $f$, the number of degrees of freedom, is $3$
for a mono-atomic income, $5$ for a diatomic income, and $6$ for an income of
arbitrarily shaped economic molecules. The quantity $R$ is the molar income constant.
For the present discussion it is important to notice that $G$ is a function of $I$ only.

In Figure 1 the economic Carnot cycle is represented in a $\bf{Q-P}$ diagram.
Consider the compression (decreasing "volume, structure, quality") of the income along the path 3-4.
Along this path the internal politics stability is $I_c$. The economic work done is
converted into production of goods $q_c$.
This process is so slow that the internal politics stability remains constant
(i.e., it is an iso-ips), and hence the
the potential growth $Q$ of the income does not change, i.e., $dQ = 0$. By the first and second law of economics,
reference to Figure 1, and the ideal income law $PQ = R\,I$ (for the sake of argument we consider one mole of income),
we have
$$q_c = I_c\int_3^4 dE = \int_3^4 PdQ = RI_c \int_3^4 \frac{dQ}{Q}\,\,\Rightarrow \,\, E_4-E_3= R \ln \,\frac{Q_4}{Q_3} < 0.$$
The economic entropy $E$ of an ideal income is a logarithmic function of its "volume, structure, quality" $Q$.

In the $Q-P$ plane, the iso-ips 3-4 has the following expression $P=\frac{R\,I_c}{Q}$,
i.e., $P$ is a function of $Q$ only.

At point 4, we move on a curve similar to an adiabatic curve.
We move inward from $Q_4$ to $Q_3$ and by this compression the income is
amplified to exactly the internal politics stability $I_h$. From 4 to 1 the
path in the $\bf{Q-P}$ diagram is similar to adiabatic and reversible path. The entropy does
not change during this compression, i.e., $dE = 0$ (an isentropic economic process).
It is easy to derive the equation for the adiabatic in the $Q-P$ plane,
$$dG=\frac{f}{2}R\,\,dI= -P\,\,dQ= -\frac{RI}{Q}\,dQ \,\,\Rightarrow \,\,\frac{f}{2}\ln \frac{I_3}{I_4} =-\ln \frac{Q_3}{Q_4}$$
and the ideal income law $PQ = RI$ (for the sake of argument we consider one mole of income) implies
$$\frac{f}{2}\ln \frac{P_3Q_3}{P_4Q_4} =-\ln \frac{Q_3}{Q_4}\,\,\Rightarrow\,\, \frac{P_3}{P_4}= \left(\frac{Q_3}{Q_4}\right)^{-\frac{f+2}{f}}.$$

For a mono-atomic income ($f=3$) the expression of the curve 4-1 is $P= c\,\,Q^{-\frac{5}{3}}$, where
$c= P_4Q_4^{\frac{5}{3}}$.

The integral $\int_3^1 P\,dQ$ is the economic work done. In the $\bf{Q-P}$ diagram this integral is
the negative of the area bounded by the curve 3-4-1 and the $Q$ axis.

At point 1, we have the internal politics stability $I_h$ and total production of goods $q_h$.
The iso-ips 1-2 has the expression $P(Q) = \frac{R\,I_h}{Q}$ and the entropy increases, i.e.,
$$E_2-E_1= R\ln \frac{Q_2}{Q_1}>0.$$

From 2 to 3 isentropic economic expansion occurs, the internal politics stability
decreases to $I_c$. The integral $\int_1^3 P\,dQ$  is the economic work.
In the $\bf{Q-P}$ diagram this integral is the area bounded by the curve 1-2-3 and the $Q$ axis.
Note that this area is larger than the area below the curve 3-4-1.
Difference of the two areas, the economic work done during the cycle, is $\oint P\,dQ$,
where the closed curve is $1-2-3-4-1$.

By the first law of economics, the total work $W$ is equal to $q_h - q_c$.

The economic process, as depicted in Figure 1, runs clockwise. All the economic processes are reversible,
so that all arrows can be reverted.

\section{Economic Van der Waals Equation}

Let $Q_m$ be the molar volume of the income, $R$ be the universal income constant,
$P$ be the price level, $Q$ be the "volume, structure, quality", and $I$ be internal politics stability.

The Van der Waals equation in Thermodynamics \cite{[21]} has a correspondent into economics,
where the variables are $P,Q,I$. The idea is based on plausible reasons that
real incomes do not follow the ideal income law.

The ideal income law states that volume $Q$ occupied by $n$ moles of any income has a
price level $P$, typically in reference currency, at internal politics stability $I$
in percents - see, for instance, Moody’s rating. The relationship for these variables,
$P Q = n R I$, is called the {\it ideal income economic law} or equation of state.

\begin{figure}
  \centering
  \includegraphics[width=12cm]{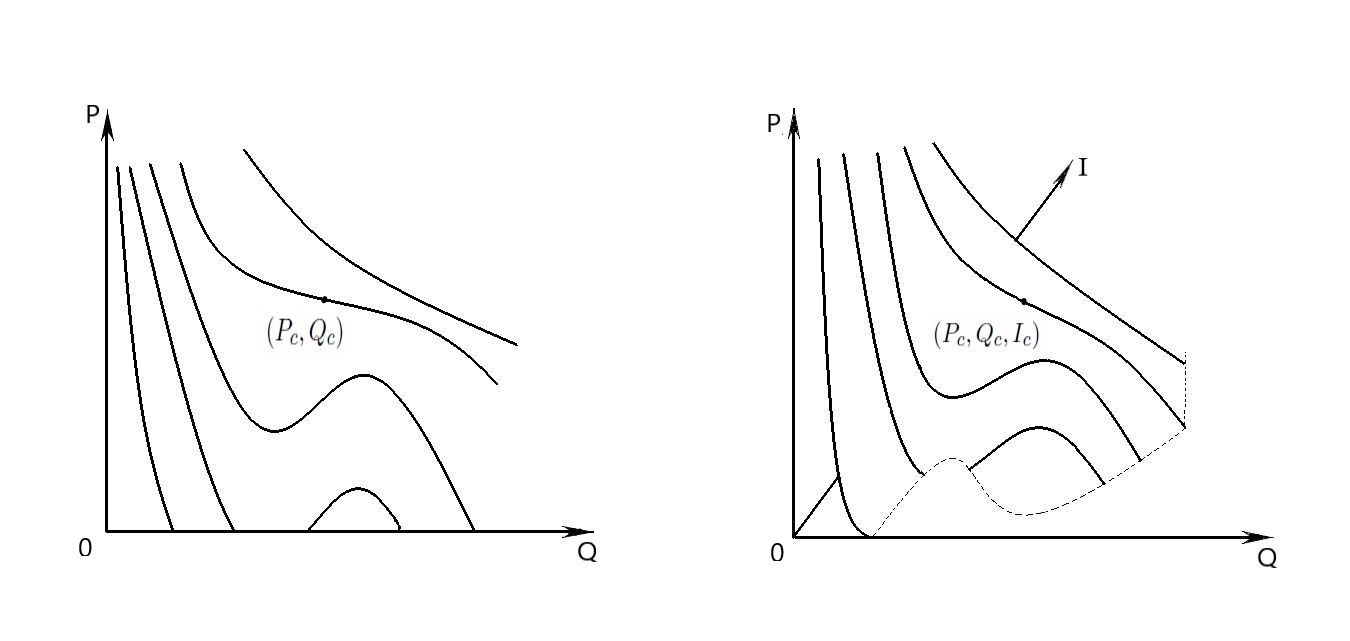}\\
 \caption{Van der Waals surface}\label{VW}
\end{figure}

The {\it economic Van der Waals equation} is
$$\left(P+\frac{a}{Q_m^2}\right)(Q_m - b)= R\,I,\,\, \hbox{or}\,\,
\left(P+\frac{an^2}{Q^2}\right)(Q - nb)= n\,R\,I.$$
When the molar volume $Q_m$ is large, the constant $b$ becomes negligible in comparison with $Q_m$,
the term $\frac{a}{Q_m^2}$ becomes negligible with respect to $P$, the economic Van der Waals
equation reduces to the ideal income economic law, $P\,Q_m = R\,I$.

The economic Van der Waals equation successfully approximates the behavior
of real "$mp-lc$ = monetary policies as liquidity or consumption"
above critical internal politics stability and is
qualitatively reasonable for points ($mp-lc$, level of prices, internal politics stability)
around critical point.
However, near the transitions between income and $mp-lc$,
in the range of $P, Q, I$, where the $mp-lc$ phase and the income phase are in equilibrium,
the Van der Waals equation fails to accurately model observed experimental behaviour;
in particular $P$ is a constant function of $Q$ at given internal politics stability,
but the Van der Waals equation do not confirm that.
As such, the Van der Waals model is not useful only for calculations intended to predict real
behavior in regions near the critical point, but also for qualitatively properties. 
Empirical corrections to address these predictive
deficiencies can be inserted into the Van der Waals model, e.g., equal area rule (see Maxwell theory),
and related but distinct theoretical models, e.g., based on the principle of
corresponding states. All of these permit to achieve better fits to real $mp-lc$
behaviour in equations of more comparable complexity.

Geometrically, the previous equation represents a surface
located in the first octane of the point space $(0,\infty)^3=\{(P, Q, I)\}$, that we
call the {\it economic Van der Waals surface}, Fig. 3. The shape of this surface follows
from the shapes of the sections by the planes $I$ = constant (Van der Waals iso-"mp-lc").
There is a stationary inflection point in the constant - "internal politics stability" curve
(critical iso-"mp-lc") on a $\bf {Q-P}$ diagram. This means that at the degenerate critical point we must have
$$\frac{\partial P}{\partial Q}\Big\vert_I =0,\,\,\frac{\partial^2 P}{\partial Q^2}\Big\vert_I =0.$$
On the Van der Waals surface, the above critical point $(P_c,Q_c,I_c)$ has the coordinates
$$P_c = \frac{a}{27b^2},\,\, Q_c = 3b,\,\, I_c= \frac{8a}{27bR}.$$

\begin{remark}
Let us show that we can find the critical point without the use of the "partial derivative".
At the critical point, the income is characterized by the critical values of
$I_c$, $P_c$, $Q_c$, which are determined only by the income properties.
From the algebraic point of view, the iso-"mp-lc" through the critical point,
i.e.,
$Q^3- \left(b + \frac{RI}{P}\right)Q^2 + \frac{a}{P}\,\,Q - \frac{ab}{P}=0,$
like equation in $Q$, has a triple real root $Q_c$, i.e., the equation should be identified to
$(Q- Q_c)^3=0$. Expanding this cube and equating the coefficients of the terms with equal powers of
$Q$,  we find the expressions for the critical parameters.
\end{remark}

Let's connect with catastrophe theory. For this we consider the function
$\varphi:(0,\infty)^3 \to \mathbb{R}^3,$ of components
$$x= \frac{1}{Q}-\frac{1}{3b},\,\, \alpha =\frac{P}{a}+ \frac{RI}{ab}-\frac{1}{3b^2},\,\, \beta = -\frac{P}{3ab}+ \frac{RI}{3ab^2}- \frac{2}{27b^3}.$$

One remarks that $\varphi:(0,\infty)^3\to \varphi((0,\infty)^3)$ is a global diffeomorphism
and $\varphi(P_c,Q_c,I_c)=(0,0,0)$. The image of economic Van der Waals surface by $\varphi$
is characterized by the equation
$x^3+\alpha x+ \beta=0,$ where
$$x \in \left(-\frac{1}{3b},\infty\right),\,\,\alpha \in \left(-\frac{1}{3b^2},\infty\right),
\,\,\beta \in \left(-\infty, \frac{-2}{27b^3}  \right) \cup \left(\frac{-2}{27b^3},\infty\right),$$
that is, it is part of the {\it cuspidal catastrophe manifold}.

\begin{theorem}
All diffeomorphic-invariant information about
Van der Waals $mp-lc$ can be obtained by examining the cuspidal potential
$ x  \to f (x, \alpha, \beta) = \frac{1}{4} x ^ 4 + \frac {1} {2} \alpha x ^ 2 + \beta x,$
which is equivalent to Gibbs' free energy.
\end{theorem}

{\bf Authors addresses}: Prof. Emeritus Dr. Constantin Udriste, Prof. Dr. Ionel Tevy,
 University  Politehnica of Bucharest, Faculty of Applied Sciences,
Department of Mathematics-Informatics, Splaiul Independentei 313, Bucharest, 060042, Romania.

Prof. Emeritus Dr. Constantin Udriste, Titular Member of Academy of Romanian Scientists,
Spl. Independentei, 54, Sector 5, Bucharest, 50085, Romania

E-mail addresses: udriste@mathem.pub.ro ; tevy@mathem.pub.ro

Prof. Dr. Vladimir Golubyatnikov, Sobolev Institute of Mathematics SB RAS
Department of Inverse and Ill-posed Problems, Kotpyug avenue 4, Novosibirsk,
630090, Russia

E-mail address: golubyatn@yandex.ru


\begin{thebibliography}{19}

\bibitem{[BDU]} C. Bianciardia, A. Donatib, S. Ulgiatic, 1993. {\it On the relationship between the economic process,
the Carnot cycle and the entropy law}, Ecological Economics, 8, 1 (1993), 7-10.


\bibitem{[1]} D. S. Chernavski, N. I. Starkov, A. V. Shcherbakov, 2002. {\it On some problems of physical economics},
Phys. Usp. 45, 9, pp. 977-997.

\bibitem{[2]} T. I. Cre\c tu, 1996. {\it Physics} (in Romanian), Technical Editorial House, Bucharest.


\bibitem{[3]} N. Georgescu-Roegen, 1971. {\it The Entropy Law and Economic Process}, Cambridge, Mass., Harvard University Press.

\bibitem{[4]} A. Holden, 1996. {\it The New Science of Complexity}, Springer.

\bibitem{[41]} W. Isard, {\it Location theory and trade theory: short-run analysis}, Quarterly Journal of Economics, 68 (2)
(1954), 305. doi:10.2307/1884452

\bibitem{[AKh]} A. Khrennikov, 2010. {\it Thermodynamic-like Approach to Complexity
of the Financial Market (in the Light of the Present Financial Crises)}, in M. Faggini, C.P. Vinci (Eds.),
Decision Theory and Choices: a Complexity Approach, 183-203, Springer-Verlag, Italia.

\bibitem{[AK]} A. Khrennikov, 2018. {\it The Financial Heat Machine: coupling with the present financial crises},
International Center for Mathematical Modelling in Physics and Cognitive Sciences, University of V\" axj\" o, S-35195, Sweden,
 e–mail: Andrei.Khrennikov@vxu.se, INTERNET.


\bibitem{[7]} S. London, F. Tohme, 2007. {\it Thermodynamics and Economic Theory: a conceptual discussion},
Asociaci\' on Argentina de Econom\' ia Pol\' itica - XXX Reuni\' on Anual, Facultad de Ciencias Econ\' omicas - Universidad Nacional de R\' io Cuarto.

\bibitem{[JM]} J. Mimkes, 2018. {\it Concepts of Thermodynamics in Economic Systems},
Physik Department, Universit\" at Paderborn, Warburgerstr. 100, D-33098 Paderborn, INTERNET,
e-mail: mimkes@physik.upb.de


\bibitem{[9]} M. Ruth, 2005, {\it Insights from thermodynamics for the analysis of economic processes}, in A. Kleidon, R. Lorenz (Eds), {\it Non-equilibrium thermodynamics and the production of entropy: life, earth, and beyond}, Springer, Heidelberg, pp. 243-254.

\bibitem{[10]} C. Udri\c ste, O. Dogaru, I. \c Tevy, 2002. {\it Extrema with Nonholonomic Constraints},
Monographs and Textbooks 4, Geometry Balkan Press.

\bibitem{[11]} C. Udri\c ste, I. \c Tevy, M. Ferrara, 2002. {\it Nonholonomic economic systems},
in C. Udri\c ste, O. Dogaru, I. Tevy, {\it Extrema with Nonholonomic Constraints},
Monographs and Textbooks 4, Geometry Balkan Press, pp. 139-150.

\bibitem{[12]} C. Udri\c ste, M. Ferrara, D. Opri\c s, 2004. {\it Economic Geometric Dynamics},
Monographs and Textbooks 6, Geometry Balkan Press.

\bibitem{13} C. Udri\c ste, 2007. {\it Thermodynamics versus Economics}, University Politehnica of Bucharest,
Scientific  Bulletin, Series A, 69, 3, pp. 89-91.

\bibitem{14} C. Udri\c ste, M. Ferrara, F. Munteanu, D. Zugr\u avescu, 2007. {\it Economics of Roegen-Ruppeiner-Weinhold type},
The International Conference of Differential Geometry and Dynamical Systems, University Politehnica of Bucharest.

\bibitem {[15]} C. Udri\c ste, M. Ferrara, 2007. {\it Multi-time optimal economic growth}, Journal of the Calcutta Mathematical Society, 3, 1, pp. 1-6.

\bibitem{[16]}  C. Udri\c ste, M. Ferrara, 2008. {\it Black hole models in economics}, Tensor,
N.S., vol. 70, 1, pp. 53-62.

\bibitem {[17]} C. Udri\c ste, M. Ferrara, 2008. {\it Multitime models of optimal growth},
WSEAS Transactions on Mathematics, 7, 1, pp. 51-55.


\bibitem {[18]} C. Udri\c ste, M. Ferrara, D. Zugr\u avescu, F. Munteanu, 2008. {\it Geobiodynamics and Roegen type economy},
Far East J. Math. Sci. (FJMS) 28, No. 3, pp. 681 - 693.

\bibitem {[UFZM]} C. Udri\c ste, M. Ferrara, D. Zugr\u avescu, F. Munteanu, 2010. {\it Nonholonomic geometry of economic systems},
ECC'10 Proceedings of the 4th European Computing Conference, Bucharest, Romania, pp. 170-177.


\bibitem{[19]} C. Udri\c ste, M. Ferrara, D. Zugr\u avescu, F. Munteanu, 2012. {\it Controllability
of a nonholonomic macroeconomic system}, J. Optim. Theory Appl., 154, 3, pp. 1036-1054.


\bibitem{[20]} C. Udri\c ste, 2013. {\it Optimal control on nonholonomic black holes},
Journal of Computational Methods in Sciences and Engineering, 13, 1-2, pp. 271-278.

\bibitem{[21]} J. D. Van der Waals, 1901-1921. {\it The equation of state for gases and liquids}. Nobel Lectures, Physics.
Amsterdam: Elsevier Publishing Company. 1967. pp. 254-265.



\end{thebibliography}
\end{document}